\documentclass[runningheads]{llncs}
\usepackage{graphicx,siunitx}

\newcommand{\Oh}[1]
	{\ensuremath{\mathcal{O}\!\left({#1}\right)}}
\newcommand{\access}
	{\ensuremath{\mathsf{access}}}
\newcommand{\rank}
	{\ensuremath{\mathsf{rank}}}
\newcommand{\select}
	{\ensuremath{\mathsf{select}}}
\newcommand{\occ}
	{\ensuremath{\mathsf{occ}}}

\newcommand{\BWT}
	{\ensuremath{\mathsf{BWT}}}
\newcommand{\C}
	{\ensuremath{\mathsf{C}}}
\newcommand{\LF}
	{\ensuremath{\mathsf{LF}}}
\newcommand{\Psiop}
	{\ensuremath{\mathsf{\Psi}}}
\newcommand{\mus}[1]
	{\SI{#1}{\micro\second}}

\begin{document}

\title{Relative Select}
\author{Christina Boucher\inst{1} \and
Alexander Bowe\inst{2} \and
Travis Gagie\inst{3} \and\\
Giovanni Manzini\inst{4} \and
Jouni Sir\'en\inst{5}}
\authorrunning{Boucher et al.}
\institute{University of Colorado, USA \and
National Institute of Informatics, Japan \and
University of Helsinki, Finland \and
University of Eastern Piedmont, Italy \and
Wellcome Trust Sanger Institute, UK}
\maketitle

\begin{abstract}
Motivated by the problem of storing coloured de Bruijn graphs, we show how, if we can already support fast select queries on one string, then we can store a little extra information and support fairly fast select queries on a similar string.
\end{abstract}

\section{Introduction}
\label{sec:introduction}

Many compressed data structures for strings rely on three fundamental queries: access, rank and select.  The query \(S.\access (i)\) on a string $S$ returns its $i$th character; the query \(S.\rank_a (i)\) returns the number of occurrences of character $a$ in the prefix of $S$ of length $i$; and the query \(S.\select_a (j)\) returns the position of the $j$th leftmost occurrence of $a$ in $S$.  Suppose we have a data structure supporting these queries on a string $S_1$ and we want another data structure supporting them on a similar string $S_2$.  It is not difficult to store $\Oh{d}$ extra words, where $d$ is the standard edit distance between $S_1$ and $S_2$ (i.e., the number of single-character insertions, deletions and substitutions needed to change one into the other), and support access to any character of $S_2$ using $\Oh{\log \log (|S_1| + |S_2|)}$ time on top of an access query on $S_1$.  Last year, when describing their relative FM-index data structure, Belazzougui et al.~\cite{BGGMS14} showed how to store $\Oh{d}$ extra words and support any rank query on $S_2$ using $\Oh{\log \log (|S_1| + |S_2|)}$ time on top of a rank query on $S_1$.  In this paper we show how to store $\Oh{d}$ extra words and support any select query on $S_2$ using $\Oh{\log \log (|S_1| + |S_2|)}$ time on top of a select query on $S_1$.  We call this {\em relative select} and we expect it to be useful when storing compressed data structures for navigating in coloured de Bruijn graphs~\cite{ICTFM12}.

Belazzougui el al.\ were interested in saving space when storing FM-indexes~\cite{FM05} for many genomes from the same species.  An FM-index for a genome is essentially just a data structure supporting access and rank on the Burrows-Wheeler Transform~\cite{BW94} (BWT) of that genome.  The BWT sorts the characters of a string into the lexicographic order of the suffixes that immediately follow them.  The edit distance between two genomes from the same species tends to be small relative to their lengths and in practice the edit distance between their BWTs also tends to be small.  Therefore, if we store the FM-index for one genome normally, we can use Belazzougui et al.'s result to save space when storing FM-indexes for other genomes from the same species (at the cost of increasing their query times).

It is possible to support nearly all the functionality of an FM-index without using select queries on the underlying BWT, so Belazzougui et al.\ did not consider relative select.  Adding it to their data structure allows us, e.g., to extract more quickly the characters following occurrences of a pattern.  Our interest in relative select, however, comes from Bowe et al.'s~\cite{BOSS12} (see also~\cite{BBGPS15}) compressed representation of de Bruijn graphs --- which is based on something like an FM-index and uses select queries to find nodes' predecessors, and which we call the BOSS representation for the authors' initials --- and the possibility of extending it to coloured de Bruijn graphs.  Our plan for future work is to view a coloured de Bruijn graph as a union of normal de Bruijn graphs, and relatively compress the BOSS representations of those graphs.  Due to space constraints, we provide a brief summary of the BOSS representation and coloured de Bruijn graphs as an appendix.  In Section~\ref{sec:design} we describe how we implement relative select, and in Section~\ref{sec:experiments} we give experimental evidence that our implementation is practical.  For simplicity, we assume throughout that the size of the alphabet is constant, and we work in the word-RAM model with \(\Omega (\log (|S_1| + |S_2|))\)-bit words.

\section{Design}
\label{sec:design}

Although our implementation of relative select is made up of steps that are individually very simple, the overall effect might be confusing.  To mitigate this, we break our presentation into pieces: first, we consider the case when $S_2$ is a subsequence of $S_1$; then, we consider the case when $S_2$ is a supersequence of $S_1$; and finally, we combine our solutions for these special cases to obtain a general solution.  We close this section with a small example.

\begin{lemma}
\label{lem:subsequence}
Given a select data structure for a string $S_1$, and a subsequence $S_2$ of $S_1$, we can store $\Oh{|S_1| - |S_2|}$ extra words and support any select query on $S_2$ using $\Oh{\log \log |S_1|}$ time on top of a select query on $S_1$.
\end{lemma}

\begin{proof}
We store a bitvector \(B [1..|S_1|]\) with 1s marking the characters of $S_1$ that do not appear in $S_2$.  For each distinct character $x$, we store a bitvector \(B_x [1..\occ (x, S_1)]\), where \(\occ (x, S_1)\) is the number of occurrences of $x$ in $S_1$, with 1s marking the occurrences of $x$ in $S_1$ that do not appear in $S_2$.  This takes a total of $\Oh{|S_1| - |S_2|}$ extra words and lets us compute
\[S_2.\select_x (i) = B.\rank_0 (S_1.\select_x (B_x.\select_0 (i)))\]
using $\Oh{\log \log |S_1|}$ time on top of a select query on $S_1$.  To see why this equality holds, consider that \(B_x.\select_0 (i)\) returns the rank in $S_1$ of the $i$th $x$ that appears in $S_2$; \(S_1.\select_x (B_x.\select_0 (i))\) returns the position of that $x$ in $S_1$; and \(B.\rank_0 (S_1.\select_x (B_x.\select_0 (i)))\) returns the position of that $x$ in $S_2$.
\qed
\end{proof}

\begin{lemma}
\label{lem:supersequence}
Given a select data structure for a string $S_1$, and a supersequence $S_2$ of $S_1$, we can store $\Oh{|S_2| - |S_1|}$ extra words and support any select query on $S_2$ using $\Oh{\log \log |S_2|}$ time on top of a select query on $S_1$.
\end{lemma}

\begin{proof}
We store a bitvector \(B [1..|S_2|]\) with 1s marking the characters of $S_2$ that do not appear in $S_1$, and a select data structure for the subsequence $D$ of $S_2$ consisting of those marked characters.  For each distinct character $x$, we store a bitvector \(B_x [1..\occ (x, S_2)]\) with 1s marking the occurrences of $x$ in $S_2$ that do not appear in $S_1$.  This takes a total of $\Oh{|S_2| - |S_1|}$ extra words and lets us compute
\[S_2.\select_x (i) =
\left\{ \begin{array}{l@{\hspace{3ex}}l}
B.\select_0 (S_1.\select_x (B_x.\rank_0 (i))) & \mbox{if \(B_x [i] = 0\),}\\
B.\select_1 (D.\select_x (B_x.\rank_1 (i))) & \mbox{if \(B_x [i] = 1\).}
\end{array} \right.\]
using $\Oh{\log \log |S_2|}$ time on top of a select query on $S_1$.  To see why this equality holds, suppose the $i$th $x$ in $S_2$ also appears in $S_1$, so \(B_x [i] = 0\).  Consider that \(B_x.\rank_0 (i)\) returns the rank of that $x$ in $S_1$; \(S_1.\select_x (B_x.\rank_0 (i))\) returns the position of that $x$ in $S_1$; and \(B.\select_0 (S_1.\select_x (B_x.\rank_0 (i)))\) returns the position of that $x$ in $S_2$.  Now suppose the $i$th $x$ in $S_2$ does not appear in $S_1$, so \(B_x [i] = 1\).  Consider that \(B_x.\rank_1 (i)\) returns the rank of that $x$ in $D$; \(D.\select_x (B_x.\rank_1 (i))\) returns the position of that $x$ in $D$; and \(B.\select_1 (D.\select_x (B_x.\rank_1 (i)))\) returns the position of that $x$ in $S_2$.
\qed
\end{proof}

\begin{theorem}
\label{thm:main}
Given a select data structure for a string $S_1$, and another string $S_2$, we can store $\Oh{d}$ extra words, where $d$ is the edit distance between $S_1$ and $S_2$, and support any select query on $S_2$ using $\Oh{\log \log (|S_1| + |S_2|)}$ time on top of a select query on $S_1$.
\end{theorem}

\begin{proof}
Consider a sequence of $d$ single-character insertions, deletions and substitutions that turns $S_1$ into $S_2$.  Let $C$ be the common subsequence of $S_1$ and $S_2$ consisting of characters left unchanged by these $d$ edits (or a longer common subsequence if we can find one).  By Lemma~\ref{lem:subsequence}, we can store $\Oh{d}$ extra words and support any select query on $C$ using $\Oh{\log \log |S_1|}$ time on top of a select query on $S_1$.  By Lemma~\ref{lem:supersequence}, we can then store $\Oh{d}$ extra words and support any select query on $S_2$ using $\Oh{\log \log |S_2|}$ time on top of a select query on $C$.  Therefore, we can store $\Oh{d}$ extra words on top of the select data structure for $S_1$ and support any select query on $S_2$ using $\Oh{\log \log (|S_1| + |S_2|)}$ time on top of a select query on $S_1$.
\qed
\end{proof}

For example, consider the strings \(S_1 = \mathsf{TCTGCGTAAAAGGTGC}\) and \(S_2 = \mathsf{TGCTCGTAAAACGCG}\) (the BWTs of {\sf GCACTTAGAGGTCAGT} and {\sf GCACTAGACGTCAGT}, respectively, from the running example in Belazzougui et al.'s paper).  Their edit distance is 5 and their longest common subsequence is \(C = \mathsf{TCTCGTAAAAGG}\).  If we already have a select data structure for $S_1$ and we want one for $S_2$, we first add support for relative select on $C$ by the bitvectors \(B, B_{\sf A}, \ldots, B_{\sf T}\), shown below; then we add support for relative select on $S_2$ by storing bitvectors \(B', B_{\sf A}', \ldots, B_{\sf T}'\), also shown below, and a select data structure for \(D = \mathsf{GCC}\).  We note that if we have a relative FM-index for $S_2$ with respect to $S_1$, then it already includes $B$, $B'$ and $D$.

\begin{center}
$\begin{array}{lcl@{\hspace{8ex}}rcl}
B [1..16] & = & 0001000000010101 & B' [1..15] & = & 010000000001010\\[1ex]
B_{\sf A} [1..4] & = & 0000 & B_{\sf A}' [1..4] & = & 0000\\
B_{\sf C} [1..3] & = & 001 & B_{\sf C}' [1..4] & = & 0011\\
B_{\sf G} [1..5] & = & 10100 & B_{\sf G}' [1..4] & = & 1000\\
B_{\sf T} [1..4]  & = & 0001 & B_{\sf T}' [1..3] & = & 000
\end{array}$
\end{center}

To compute \(S_2.\select_{\sf C} (4)\), for instance, we check \(B_{\sf C}' [4]\) and see it is 1, meaning the fourth {\sf C} in $S_2$ does not appear in $C$.  Since \(B_{\sf C}'.\rank_1 (4) = 2\), it is the second {\sf C} in $D$.  Since \(D.\select_{\sf C} (2) = 3\), it is the third character in $D$.  Finally, since \(B_1'.\select_1 (3) = 14\), it is the 14th character in $S_2$, meaning \(S_2.\select_{\sf C} (4) = 14\).

To compute \(S_2.\select_{\sf G} (3)\), we check \(B_{\sf G}' [3]\) and see it is 0, meaning the third {\sf G} in $S_2$ also appears in $C$.  Since \(B_{\sf G}'.\rank_0 (3) = 2\), it is the second {\sf G} in $C$.  Since
\[C.\select_{\sf G} (2) = B.\rank_0 (S_1.\select_{\sf G} (B_{\sf G}.\select_0 (2))) = 11\,,\]
it is the 11th character in $C$.  Finally, since \(B_1'.\select_0 (11) = 13\), it is the 13th character in $S_2$, meaning \(S_2.\select_{\sf G} (3) = 13\).

\section{Experiments}
\label{sec:experiments}

We augmented the existing implementation of the Relative FM-index with our new select structure. The implementation is written in C++ and based on the Succinct Data Structures Library 2.0~\cite{GBMP14}. We used g++ version 4.8.1 to compile the code. Our experiments were run in a computer cluster with two 16-core AMD Opteron 6378 processors in each node. The nodes were running Linux kernel 2.6.32. Query tests were run on a single core in a dedicated node with no other load.

As our reference sequence, we chose the 1000 Genomes Project's version of the GRCh37 human reference genome, both with (3.096 Gbp) and without (3.036 Gbp) chromosome Y. For a target sequence, we chose the maternal haplotype of the 1000 Genomes Project's individual NA12878 (3.036 Gbp)~\cite{Rozowsky2011}. We built a plain FM-index for the reference sequences and the target sequence, as well as relative FM-indexes for the target sequence relative to both references and with and without structures for relative select; the lengths of the common subsequences used were 2.992 Gbp and 2.991 Gbp, respectively. In all cases, we used plain bitvectors in the wavelet trees and entropy-compressed bitvectors~\cite{RRR07} for marking the common subsequences.

To test the performance of relative select, we ran 100 million random $\Psiop(i) = \BWT.\select_{c}(i - \C[c])$ queries on the BWT of the target sequence, using a plain FM-index and Relative FM-indexes with and without relative select. (Character $c$ is the $i$th character in the BWT in sorted order, while $\C[c]$ is the number of occurrences of characters smaller than $c$ in the BWT.) The implementation of $\Psiop$ in the Relative FM-index without relative select was based on binary searching with rank queries. As a comparison, we also ran $\LF(i) = \C[\BWT[i]] + \BWT.\rank_{\BWT[i]}(i)$ queries. Table~\ref{table:queries} shows the results: the relative FM-indexes without relative select are each about a fifth the size of the normal FM-indexes but rank queries are about seven times slower and select queries are about forty times slower; the relative FM-indexes with relative select are about a third the size of the normal FM-indexes but select queries are only about five times slower (rank queries are unaffected).

\begin{table}[tb!]
\caption{Average query times for 100 million random $\LF$ and $\Psiop$ queries on NA12878 stored relative to the human reference genome, with and without chromosome Y.}
\label{table:queries}
\setlength{\tabcolsep}{3pt}
\resizebox{\textwidth}{!}
{\begin{tabular}{ccccccccc}
\hline
\noalign{\smallskip}
 & \multicolumn{3}{c}{\textbf{FM-index}} & \multicolumn{3}{c}{\textbf{Relative FM-index}} & \multicolumn{2}{c}{\textbf{$+$\,Relative Select}} \\
\textbf{ChrY} & \textbf{space} & \textbf{$\LF$} & \textbf{$\Psiop$} & \textbf{space} & \textbf{$\LF$} & \textbf{$\Psiop$} & \textbf{total space} & \textbf{$\Psiop$} \\
\noalign{\smallskip}
\hline
\noalign{\smallskip}
yes & 1090 MB & \mus{0.55} & \mus{1.22} & 218 MB & \mus{3.95} & \mus{48.0} & 382 MB & \mus{6.11} \\
no  & 1090 MB & \mus{0.55} & \mus{1.11} & 181 MB & \mus{3.84} &  \mus{44.8} & 331 MB & \mus{6.12} \\
\noalign{\smallskip}
\hline
\end{tabular}}
\end{table}

\bibliographystyle{splncs03}
\bibliography{select}

\appendix

\section{de Bruijn Graphs}
\label{sec:graphs}

In biology, the (edge-centric) {\em $k$th-order de Bruijn graph} for a set of strings (e.g., DNA reads) is the graph whose nodes are those strings' $k$-mers (substrings of length $k$), with a directed edge \((u, v)\) from $u$ to $v$ if at least one of the strings contains a corresponding substring of length \(k + 1\) with $u$ as a prefix and $v$ as a suffix.  We label \((u, v)\) with the last character of $v$.  Almost all state-of-the-art DNA assemblers build contigs via Eulerian assembly~\cite{IW95,PTW01} on de Bruijn graphs, making their space- and time-efficient representation an important problem in bioinformatics.

Bowe et al.\ add certain dummy nodes and edges, sort the edges into the right-to-left lexicographic order of the nodes they leave, and take the last column of the matrix whose rows are the edges in sorted order (or, equivalently, take the last character in each edge).  The result is like a BWT in which edges correspond to characters and nodes correspond to the substrings containing all their out-edges' characters.  For example, for the string {\sf TACGTCGACGACT} and \(k = 3\), Bowe et al.\ add nodes {\sf \$\$\$}, {\sf \$\$T} and {\sf \$TA} and edges {\sf \$\$\$T}, {\sf \$\$TA} and {\sf \$TAC} to obtain the graph shown on the right side of Figure~\ref{fig:graph}; build the matrix shown on the left side of the figure; and take the last column {\sf TCCGTGGATAA\$C}.  (This example is from~\cite{BBGPS15}.)  With some auxiliary data structures, we can use rank and select queries on this edge-BWT to navigate forward and backward in the graph.

\begin{figure}[t!]
\resizebox{\textwidth}{!}
{\begin{tabular}{c@{\hspace{6ex}}c}
\raisebox{-10ex}
{\includegraphics*[trim = 0cm 0cm 9cm 24cm, width=50ex]{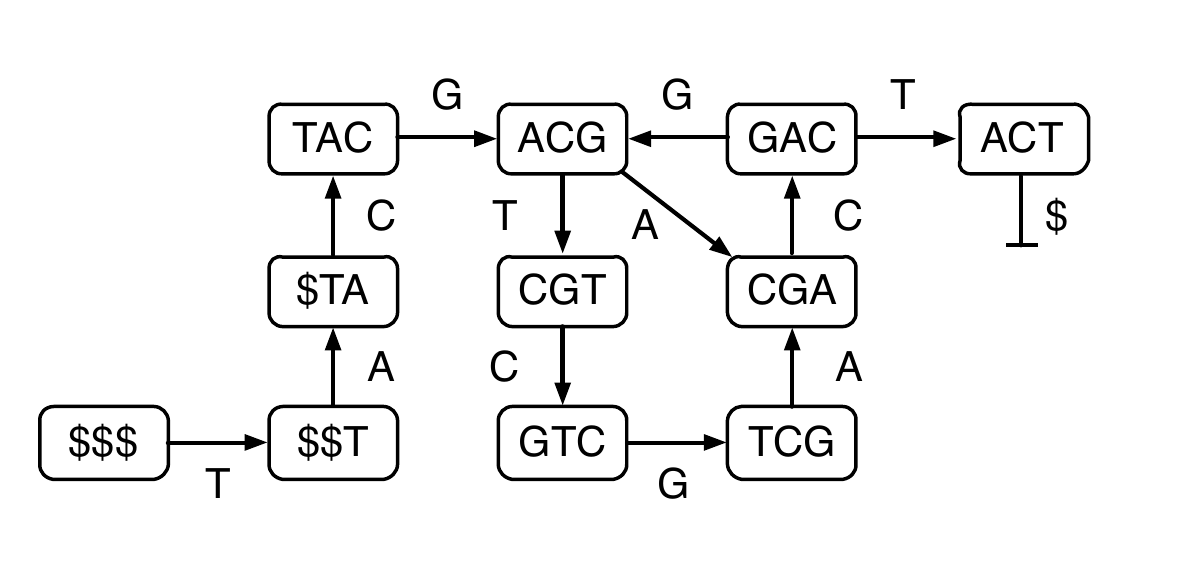}} &
\begin{tabular}{r@{\hspace{1ex}}@{\hspace{1ex}}@{\hspace{1ex}}l@{\hspace{1ex}}c
	@{\hspace{4ex}}r@{\hspace{1ex}}@{\hspace{1ex}}@{\hspace{1ex}}l@{\hspace{1ex}}c}
1) & \sf \,\$\,\$\,\$\, & \sf T & 8) & \sf ACG & \sf A\\
2) & \sf CGA & \sf C & 9) & \sf ACG & \sf T\\
3) & \sf \,\$\,TA & \sf C & 10) & \sf TCG & \sf A\\
4) & \sf GAC & \sf G & 11) & \sf \,\$\,\$\,T & \sf A\\
5) & \sf GAC & \sf T & 12) & \sf ACT & \sf \$\\
6) & \sf TAC & \sf G & 13) & \sf CGT & \sf C\\
7) & \sf GTC & \sf G & & &
\end{tabular}
\end{tabular}}
\caption{Bowe et al.'s augmented de Bruijn graph (left) and matrix (right) for the string {\sf TACGTCGACGACT}; the last column {\sf TCCGTGGATAA\$C} is like a BWT of the edges.}
\label{fig:graph}
\end{figure}

For the two strings {\sf TACGTCGACGACT} and {\sf TACGACGCGACT} and \(k = 3\), the de Bruijn graph is 2 nodes larger than the graphs for strings separately.  If we store whether each edge occurs in the first string, the second string, or both, then the result is a {\em coloured de Bruijn graph}.  Coloured de Bruijn graphs were introduced by Iqbal et al.~\cite{ICTFM12} for detecting variations between individuals' genomes, and are now also used in other areas of genomics (see, e.g.,~\cite{BP15}).  We can view the coloured de Bruijn graph as the union of each graph consisting of edges of the same colour.  In a future paper we will show how to combine the BOSS representations of the individual de Bruijn graphs to obtain a representation of the coloured de Bruijn graph, and also how to relatively compress the auxiliary data structures for the BOSS representations of the individual graphs.

We can use Belazzougui et al.'s result to relatively compress the edge-BWTs of the individual graphs while still supporting rank over them.  For example, the edge-BWTs for {\sf TACGTCGACGACT} and {\sf TACGACGCGACT} with \(k = 3\) are {\sf TCCGTGGATAA\$C} and {\sf TCCGTGGACAA\$}, respectively.  They are so close --- edit distance 2 --- because most of the strings' 4-tuples are common to both and, thus, most of their de Bruijn graphs' edges are common to both.  We note that, for reasonable values of $k$, most of the \((k + 1)\)-mers in genomes from the same species should also be common to most of the genomes.  In this paper we showed how to support relative select on similar strings, which we will eventually need to navigate backward across edges in our representation of coloured de Bruijn graphs.

\end{document}